\documentclass{amsart}
\usepackage{graphicx}
\usepackage{amsmath,amssymb,latexsym} % ams auxiliaries.
\usepackage{etoolbox}
\usepackage{enumitem}
\usepackage{color}
\usepackage{tikz}
\usetikzlibrary{arrows,shapes, calc, fit, positioning}
\usepackage{caption}
\usepackage[colorlinks,linkcolor=blue,citecolor=magenta]{hyperref}
\usepackage{subcaption}
\usepackage[ruled,vlined,linesnumbered]{algorithm2e}

\SetCommentSty{mycommfont}
\usepackage{mathtools}
\usepackage{todonotes}
\usepackage[left=1.2in,right=1.2in,top=1in,bottom=1in]{geometry}
\usepackage[normalem]{ulem}
\usepackage{url}
\usepackage{natbib}
\usepackage{booktabs}

\theoremstyle{plain}
\newtheorem{theorem}{Theorem}
\newtheorem{lemma}[theorem]{Lemma}

\newtheorem{proposition}[theorem]{Proposition}

\theoremstyle{definition}

\usepackage{tikz,enumitem,framed}

\def\le{\leqslant}
\def\ge{\geqslant}

\def\Z{\mathbb{Z}}

\def\R{\mathbb{R}}

\def\com{\operatorname{Cof}}

\begin{document}
\title[]{Fair and efficient allocation of indivisible items under category constraints}
%\footnote{A preliminary version of this paper appeared in the Proceedings of the 17th Conference on Web and Internet Economics (WINE 2021). This version extends all the previous results to the group version and contains a new section with several results for proportionality (Section~\ref{sec:prop}).}
\author{Ayumi Igarashi and Fr\'ed\'eric Meunier}
\address{University of Tokyo, Tokyo, Japan}
\email{igarashi@mist.i.u-tokyo.ac.jp}
\address{\'Ecole Nationale des Ponts et Chauss\'ees, France}
\email{frederic.meunier@enpc.fr}

%\address[S. Iemoto]{Department of Math ematical and Computing Sciences, Tokyo Institute of Technology, Oh\-oka\-yama, Me\-guro\-ku, Tokyo 152-8552, Japan}
%\email{igarashi@mist.i.u-tokyo.ac.jp}
%\email{frederic.meunier@enpc.fr}
%\subjclass[2000]{??}

\maketitle

\begin{abstract}
We study the problem of fairly allocating indivisible goods and chores under category constraints. Specifically, there are $n$ agents and $m$ indivisible items which are partitioned into categories with associated capacities. An allocation is considered feasible if each bundle satisfies the capacity constraints of its respective categories. For the case of two agents, Shoshan et al. (2023) recently developed a polynomial-time algorithm to find a Pareto-optimal allocation satisfying a relaxed version of envy-freeness, called EF$[1,1]$. Extending such guarantees beyond two agents has remained open.

We make progress toward this question by proving that for any number $n$ of agents, there always exists a Pareto-optimal allocation in which each agent can be made envy-free by reallocating at most $\min \{k+1,n\}(n-1)$ items. Moreover, when the number of agents is constant, we give a polynomial-time algorithm to compute such an allocation. Our results rely on a new application of the Knaster–Kuratowski–Mazurkiewicz (KKM) lemma to a simplex of agent weights, which may be of independent interest for fair division problems involving indivisible items.
\end{abstract}

\section{Introduction}\label{sec:intro}

\subsection{Fair and efficient allocations of indivisible items} 
Fair allocation of resources is a fundamental societal problem that has attracted significant attention from the economics literature. The challenge of balancing efficiency and fairness in resource allocation has been one of the central focuses in the literature. In the context of divisible resources, these seemingly conflicting objectives can, in fact, be compatible. For a broad class of utility functions, an allocation that satisfies both Pareto-optimality as an efficiency criterion and envy-freeness as a fairness criterion is guaranteed to exist~\citep{Varian1976}. Moreover, for linear non-negative utilities, such an allocation can be achieved by maximizing the \emph{Nash welfare}, i.e., the product of agents' utilities, as shown by \citet{eisenberg}.

In the allocation of indivisible resources---such as assigning time slots to employees or courses to students---envy-freeness is not always attainable. To address this limitation, \citet{budish2011combinatorial} introduced the notion of EF1, a relaxation of envy-freeness. A striking result, recently proven by \citet{CaragiannisKuMo19}, establishes a discrete analogue of Eisenberg and Gale: in the case of items yielding non-negative utilities (i.e., goods), an indivisible allocation that maximizes the Nash welfare satisfies both Pareto-optimality and EF1, provided that agents have additive utilities. Further, such an allocation can be found in pseudo-polynomial time~\citep{BarmanKrVa18}. 

In many real-world applications, allocations must account for constraints. 
A particularly relevant class of constraints is \emph{category constraints}, recently studied by \citet{DrorFeSe23} and \citet{ShoshanSeHa23}, where items are partitioned into disjoint categories, and each agent can receive at most a certain number of items from each category. 

As an illustration, consider the allocation of teaching duties among faculty members within a department, where courses are grouped into categories such as undergraduate, graduate, and summer-school courses. To ensure a balanced workload, a university may impose category constraints; for example, each faculty member may be assigned at most one undergraduate course and at most one graduate course per year.

Similar constraints arise in many other allocation scenarios. In the allocation of time slots across different activities, even if participants prefer multiple slots from the same activity, institutional rules often impose upper bounds on how many slots of each activity type a participant can receive. Likewise, in the allocation of medical resources among municipalities, resources may be divided into categories such as equipment, vehicles, and staff. Capacity and operational limitations then impose upper bounds on the number of resources from each category that a municipality can receive, naturally giving rise to category constraints.

The addition of such constraints increases the complexity of the allocation problem, making it more challenging to achieve both efficiency and fairness. \citet{ShoshanSeHa23} provided an example where no allocation satisfies both Pareto-optimality and EF1 under category constraints when agents have positive or negative utilities for the items. Nevertheless, they demonstrated that a feasible allocation satisfying Pareto-optimality and EF$[1,1]$ (a relaxation of EF1) exists and can be efficiently computed for two agents.

In this paper, we study fair allocation of indivisible items under category constraints. Our goal is to find a feasible allocation that is simultaneously efficient and approximately fair.  
Specifically, we build on the setting of Shoshan, Hazon, and Segal-Halevi. 
There are $n$ agents and $m$ items. Each agent $i \in [n]$ has a utility function $u_i \colon [m] \rightarrow \Z$ over the items. For a set $S$ of items, we write $u_i(S)=\sum_{j \in S}u_i(j)$. We note that our setting allows for the allocation of both goods and chores, meaning that agents may have positive or negative utilities for individual items.
An \emph{allocation} $A=(A_1,A_2,\ldots,A_n)$ is an ordered partition of the items, i.e., $A_i \cap A_{i'}=\varnothing$ for every distinct pair of agents $i,i'$ and  $\bigcup_{i \in [n]}A_i=[m]$. The items are moreover partitioned into categories $S_h$ associated with capacities $s_h$ for $h \in [k]$. An allocation $A$ is {\em feasible} if $|A_i \cap S_h| \le s_h$  for every agent $i$ and category $h$. We assume that
$|S_h| \le n s_h$ for every $h \in [k]$ to ensure the existence of at least one feasible allocation. 
An allocation~$A$ is \emph{Pareto-optimal} if it is feasible and if there is no feasible allocation $A'$ such that $u_i(A'_i) \ge u_i(A_i)$ for every agent $i$ and $u_{i'}(A'_{i'}) > u_{i'}(A_{i'})$ for at least one agent $i'$. An agent $i$ \emph{envies} another agent $i'$ if $u_i(A_{i'}) > u_i(A_i)$.

\subsection{Main results}
The case $n=2$ of the following theorem is the main result of Shoshan, Hazon, and Segal-Halevi. We generalize this result for all $n$.

\begin{theorem}\label{thm:main}
For every number $n$ of agents, there exist a subset of $\min \{k+1,n\} (n-1)$ items and a Pareto-optimal allocation such that for every agent, it is possible to reallocate some items in the subset so as to get a Pareto-optimal allocation making the agent non-envious.
\end{theorem}

An allocation $A$ satisfies \emph{envy-freeness up one good and one chore} (EF$[1,1]$) if for every pair of agents $i,i'$, there exist sets $S \subseteq A_i$ and $T \subseteq A_{i'}$ of size at most $1$ such that $u_i(A_i \setminus S) \ge u_i(A_{i'} \setminus T)$. Theorem~\ref{thm:main} implies the result of \cite{ShoshanSeHa23}, which states the existence of a Pareto-optimal and EF[1,1] allocation among two agents, as we explain now. We assume without loss of generality that $|S_h|=n s_h$ for every $h \in [k]$. Indeed, as noted by~\cite{ShoshanSeHa23}, this can be ensured by introducing sufficiently many dummy items for which each agent has utility $0$. Suppose that there are two Pareto-optimal allocations $A$ and $A'$ for two agents that can be transformed into each other by reallocating two items, and each agent is envy-free in one of the allocations, as ensured by Theorem~\ref{thm:main}. 
Since each agent receives exactly the same number $s_h$ of items from each category $h$ under $A$ and $A'$, these allocations can be obtained from each other by swapping two items. Namely, there exist two items $j,j'$ and two sets $I_1,I_2$ of items such that $A_1=I_1 \cup \{j\}$, $A_2 =I_2 \cup \{j'\}$, $A'_1=I_1 \cup \{j'\}$, and $A'_2 =I_2 \cup \{j\}$. By Pareto-optimality, at least one of the inequalities $u_1(I_1) \ge u_1 (I_2)$ and $u_2(I_2) \ge u_2(I_1)$ holds. Without loss of generality, assume that the first inequality holds. Then, it follows that the allocation $A$ or $A'$ that makes agent $2$ envy-free is both Pareto-optimal and EF$[1,1]$ for the two agents. 

We also extend the algorithm result of Shoshan, Hazon, and Segal-Halevi.

\begin{theorem}\label{thm:algo}
    The subset of items and the allocation of Theorem~\ref{thm:main} can be computed in polynomial time when $n$ is a constant.
\end{theorem}

\subsection{Proof techniques}

Before explaining our proof idea, let us first describe the approach by \cite{ShoshanSeHa23}. 
In the case of two agents, \cite{ShoshanSeHa23} search for 
allocations that maximize $\sum_{i=1,2} t_i u_i(A_i)$ (weighted utilitarian optimal allocations) and update the weights along $t_1 \in (0,1)$, with $t_1 + t_2 = 1$.
If there are multiple optimal allocations, they break ties arbitrarily. 
Due to the classical characterization that fractional Pareto-optimal allocations are equivalent to weighted utilitarian optimal allocations with positive weights $t_i$~\citep{Negishi1960}, %,Varian1976
the method generates a sequence of Pareto-optimal allocations, where each consecutive allocation is obtained by swapping two items. Along the way, there is a switching point at which the envy-relation reverses direction, and this point can be interpreted as the desired allocation.

%In the above approach, moving from one Pareto-optimal allocation to the next one because of the change in the weights is reminiscent to the pivot step in the ``shadow vertex'' version of the simplex algorithm. This point of view was somehow the starting point of our method, even if this is maybe not transparent anymore in the final version. Yet, we kept the idea of relying to linear programming.
In the above approach, moving from one Pareto-optimal allocation to another induced by changes in the weights is reminiscent of the pivot steps in the ``shadow vertex'' version of the simplex algorithm.
This perspective was one of the starting points of our method, even though it may no longer be fully transparent in the final presentation. Nevertheless, we retain the core idea of relying on linear programming.

To generalize the result of \cite{ShoshanSeHa23} to the case of $n$ agents, we consider weighted utilitarian optimal allocations, where the weights lie now in the $(n-1)$-dimensional standard simplex $\Delta^{n-1}$. For weights on the boundary of the simplex, there exists an agent in the positive support who is envy-free at the corresponding optimal allocation. Then, applying the KKM (Knaster–Kuratowski–Mazurkiewicz) lemma, which is a set-covering variant of Brouwer's fixed-point theorem and Sperner's lemma, demonstrates the existence of a weight such that, for each agent $i$, there is an optimal allocation that makes agent $i$ envy-free.

However, a direct application of the above approach does not necessarily yield Theorem~\ref{thm:main}. One difficulty is that weighted utilitarian optimal allocations, where the weights lie on the boundary of the simplex, are not necessarily Pareto-optimal. Another issue is that the optimal solutions may form a high-dimensional object (actually, a polytope), which could require reallocating many items to transition from one optimal allocation to another for the given weight. To address these issues, we slightly shrink the simplex to ensure that all weights are positive, while preserving the covering condition of the KKM lemma. This is actually not done explicitly but with a carefully-chosen expression for the weights that mimics exactly a shrinking of $\Delta^{n-1}$. Further, we perturb our objective function to reduce the dimension of each optimal face, establishing Theorem~\ref{thm:main}. The number of optimal faces we consider is bounded by a polynomial in the input size when $n$ is a constant, yielding Theorem~\ref{thm:algo}. %\ayumi{Note that since the Pareto-optimal allocation guaranteed in Theorem~\ref{thm:main} maximizes the weighted utilitarian social welfare $\sum^n_{i=1} t_i u_i(A_i)$ for positive weights $t_i$, it is in fact \emph{fractionally Pareto-optimal}, meaning that there is no fractional allocation that weakly improves every agent’s utility and strictly improves at least one agent’s utility.}
%\todo[inline]{Our approach may not be fPO because of the second term in the objective function.}

\subsection{ Related work}
  
For two agents, \citet{AzizCaIg22} designed an algorithm to find a Pareto-optimal and EF1 allocation of goods and chores, based on the classical Adjusted Winner algorithm~\citep{BramsTa96}. An allocation $A$ satisfies \emph{envy-freeness up to one item} (EF1) if for every pair of agents $i,i'$, there exists a set $S \subseteq A_i \cup A_{i'}$ of size at most $1$ such that $u_i(A_i \setminus S) \ge u_i(A_{i'} \setminus S)$.
\citet{ShoshanSeHa23} extended this result to a constrained setting, designing an algorithm to find a Pareto-optimal and EF$[1,1]$ allocation of goods and chores among two agents. Note that EF$[1,1]$ reduces to EF1 when all items are goods or when all items are chores. 
\citet{DrorFeSe23} considered the problem of finding an EF1 allocation under heterogeneous matroid constraints. They observed that unlike the unconstrained setting, maximizing the product of utilities may fail to satisfy EF1 even when two agents have binary additive utilities; see Example 3 in their paper. A similar example actually shows that maximizing Nash welfare may fail to achieve EF$k$, defined as an allocation where the envy is required to disappear after removing at most $k$ items; see the construction in the proof of Theorem~5 by \citet{Cookson2025}. 
\citet{Cookson2025} showed that maximizing Nash welfare satisfies $1/2$-EF1 and Pareto-optimality for non-negative utilities under any matroid constraints, where in this setting, an allocation $A$ satisfies $1/2$-EF1 if for every pair of agents $i,i'$, there exists a set $S \subseteq A_{i'}$ of size at most $1$ such that $u_i(A_i) \ge \frac{1}{2} u_i(A_{i'} \setminus S)$. \citet{bogomolnaia2024teams} considered a constrained setting in which each agent must receive exactly the same number of items, and proposed a fairness notion called EFf1 (envy-freeness up to the exchange of a single item), which requires that an agent does not envy another after exchanging a single item between their respective bundles. They examined the relationships between different fairness notions, showing that EF1 does not imply EFf1 and vice versa. They further established the compatibility of Pareto-optimality and EFf1 in several special cases, including the two-agent setting, instances with identical utilities, and instances with binary utilities. 
See \citet{Suksompong21} for other types of constraints relevant to fair division. 

\subsection{Independent and subsequent related work}
Since the initial release of the paper~\citep{arXiv}, several works have applied the same technique---namely, using the KKM lemma on a simplex of agent weights---to obtain desirable allocations. Independently of our work, \citet{barman2025mixedmanna} applied the KKM lemma to the allocation of goods and chores in the unconstrained setting, using perturbations of the utility functions to ensure an acyclic structure of optimal allocations. Similar to our work, \citet{barman2025mixedmanna} established the existence of Pareto-optimal allocations that can be made envy-free by reallocating a small number of items, as well as a polynomial-time algorithm to compute such allocations when the number of agents is a constant.
Building on this line of research, \citet{mahara2026chores}\footnote{Note that \citet{mahara2026chores} was developed independently of \citet{barman2025mixedmanna}.} established the existence of Pareto-optimal and EF1 allocations of indivisible chores, thereby resolving a major open question in fair division. His paper developed a procedure that converts allocations guaranteed by a KKM point into a Pareto-optimal and EF1 allocation by exploiting the dual program of weighted utilitarian maximization. \citet{barman2025proximately} further investigated a related approach in the mixed setting. 
We note that all of the above results focus on unconstrained settings, whereas our work considers a more general model with category constraints.

\section{Technical preliminaries of linear programming}\label{sec:prelim}
We begin by introducing some elementary facts from linear programming, along with the notions of a feasible basis and reduced costs. Complementary material about linear programming can be found in the many textbooks about linear programming, such as the one by \cite{matouvsek2007understanding}.
Consider a linear program in the standard form
\[
\begin{array}{rll}
  \text{\textup{maximize}}  & c \cdot x \\ \text{\text{\textup{s.t.}}} & Mx = b \\ & x \in \R_+^d \, ,
\end{array}
\]
where $M$ is assumed to be of full row rank. A {\em basis} $B$ is a subset of columns of $M$ such that the matrix $M_B$ is non-singular. (We use the standard notation of linear programming where a matrix with a subset of columns as a subscript is the matrix restricted to those columns. Similarly, a vector with a subset of entries as a subscript is the vector restricted to those entries.) It is {\em feasible} if $M_B^{-1}b \ge 0$. By a straightforward computation, for any basis $B$ and its complement $N$ in the set of columns, the objective function of the program can be equivalently written as $c_B^\top M_B^{-1}b + (c_N^\top - c_B^\top M_B^{-1}M_N)x_N$ when $x$ is a feasible solution. The components of $c_N^\top - c_B^\top M_B^{-1}M_N$ are the {\em reduced costs} associated with $B$. It is well-known that if the program admits an optimal solution, then there exists a feasible basis $B^*$ such that there is an optimal solution with support included in $B^*$ and such that the associated reduced costs are all non-positive. What is also well-known is that the set of optimal solutions forms a face of the set of feasible solutions, which is a polyhedron. There is actually a relation between the dimension of the optimal face and the reduced costs, as shown by the following lemma. This is ``common-knowledge'' in linear programming, but we provide a complete proof since we are not aware of a proper reference.

\begin{lemma}\label{lem:face-dimension}
    The dimension of the optimal face is at most the number of $0$'s among the reduced costs associated with $B^*$.
\end{lemma}

\begin{proof}
    Denote by $q$ the dimension of the optimal face, and denote by $y$ the optimal solution whose support is included in $B^*$. There exist $q$ points $z^1,z^2,\ldots,z^q$ in the face such that the $z^i - y$ are linearly independent. Since $M_{B^*}$ is non-singular, each of those $z^i$ has at least one non-zero component in $N^*$ (complement of $B^*$): otherwise, $z_{B^*}^i$ would be a solution of $M_{B^*}z_{B^{*}}^i = b$. Denote by $H$ the subset of columns formed by %the support of $y$ 
    $B^*$ together with the supports of the $z^i$. The $z_H^i-y_H$ being linearly independent vectors in the kernel of $M_H$, the corank of this latter matrix is lower-bounded by $q$ and the rank-nullity theorem implies $|H| \ge |B^*| + q$. By the optimality of the $z^i$, the reduced costs indexed by elements in $H \setminus B^*$ are equal to $0$. Combining this with the previous inequality leads to the desired conclusion. 
\end{proof}

As explained in the introduction, keeping the dimension of the set of optimal solutions low is crucial, and Proposition~\ref{prop:dimlp} below ensures that this can be achieved by perturbing the objective function with arbitrarily small values $\varepsilon_{s}$.
%choosing arbitrarily small $\varepsilon_{ij}$. 
(The proof of Proposition~\ref{prop:dimlp} actually shows that the set of optimal solutions of this relaxation is generically a polytope of low dimension with respect to the $\varepsilon_{s}$.) 

\begin{proposition}\label{prop:dimlp}
Let $c^1,c^2,\ldots,c^q$ be vectors in $\R^d$, and let $M$ be an $r \times d$ real matrix and $b$ be a vector of $\R^r$. Then we can choose arbitrarily small $\varepsilon_s>0$ for $s \in [d]$ such that the dimension of the optimal face of the linear program
\begin{equation}\label{eq:shrink}
\begin{array}{rll}
  \text{\textup{maximize}}  & c(t') \cdot x \\ \text{\text{\textup{s.t.}}} & Mx = b \\ & x \in \R_+^d \, ,
\end{array}
\end{equation}
is at most $q-1$ for every $t'$ in the $(q-1)$-dimensional standard simplex $\Delta^{q-1}$, where
$c_s(t') = \sum_{i=1}^q t'_i c^i_s + \varepsilon_s$. 
\end{proposition}

\begin{proof}
    Without loss of generality, we can assume $M$ to be of full row rank and that $d \ge q+r$ (since otherwise the set of feasible solutions itself is of dimension at most $q-1$). For every pair $(B,S)$ where $B$ is a feasible basis and $S$ a size-$q$ subset of $[d]\setminus B$, we introduce the following $q\times q$ matrix
    \[
    P(B,S) \coloneqq C_S - C_B M_B^{-1}M_S \, ,
    \]
    where $C$ is the $q\times d$ matrix whose $i$th row is the vector $c^{i}$. For each pair $(B,S)$ for which $P(B,S)$ is non-singular, set $\lambda(B,S) \coloneqq P(B,S)^{-1}\boldsymbol{1}$, where $\boldsymbol{1}$ is the all-one vector. For each pair $(B,S)$ for which $P(B,S)$ is singular, choose a non-zero vector $\lambda(B,S)$ in $\R^S$ such that $P(B,S)\lambda(B,S) = 0$.

    Sample now each $\varepsilon_s$ independently and uniformly at random on the open interval $(0,\alpha)$, for an arbitrary $\alpha >0$. Remark the following two facts:
    \begin{enumerate}[label=(\roman*)]
        \item\label{nonone} For the pairs $(B,S)$ for which $P(B,S)$ is non-singular, the relation
        $\lambda(B,S)^\top \bigl(M_S^\top M_B^{-1\top} \varepsilon_B - \varepsilon_S\bigl) \neq 1$ holds almost surely.
        \item\label{nonzero} For the pairs $(B,S)$ for which $P(B,S)$ is singular, the relation $\lambda(B,S)^\top \bigl(M_S^\top M_B^{-1\top} \varepsilon_B - \varepsilon_S\bigl) \neq 0$ holds almost surely.
    \end{enumerate}
    Indeed, at least one component of $\lambda(B,S)$ is non-zero and the corresponding component of $\varepsilon_S$ is drawn independently of the other components. So, from now on, we can assume that the $\varepsilon_s$ have been chosen so that the relations of items~\ref{nonone} and~\ref{nonzero} are satisfied for all pairs $(B,S)$.

    Let $t' \in \Delta^{q-1}$, and let $B^*$ be an optimal basis of the linear program. It can be chosen so that each reduced cost with respect to $B^*$ is non-positive. Suppose for a contradiction that the optimal face is of dimension at least $q$. According to Lemma~\ref{lem:face-dimension}, we can find $q$ reduced costs equal to $0$, i.e., a subset $\overline{S}$ of $q$ indices in $[d]\setminus B^*$ such that ${c_{\overline{S}}(t')}{^\top} = {c_{B^*}}(t'){^{\top}}M_{B^*}^{-1}M_{\overline{S}}$. By definition of $c(t')$, this equality writes
    \begin{equation}\label{eq:reducedsbar}
    P(B^*,\overline{S})^\top t' = M_{\overline{S}}^\top M_{B^*}^{-1\top} \varepsilon_{B^*} - \varepsilon_{\overline{S}} \, .\
    \end{equation}
    Multiplying by $\lambda(B^*,\overline{S})^\top$ this equality, we get a contradiction with the relations of items~\ref{nonone} and~\ref{nonzero} and above. (We use $\sum_{i=1}^qt'_i=1$ when $P(B^*,\overline{S})$ is non-singular.) 
\end{proof}

We emphasize that controlling the dimension of the optimal face uniformly with respect to all weights $t'_i$ is different from the standard problem of perturbing a linear program. In the usual perturbation setting, one modifies the objective or constraints to ensure uniqueness or non-degeneracy of an optimal solution. In contrast, here we aim to guarantee structural properties of the entire optimal face simultaneously for a family of weight vectors, which requires a different line of argument.

The proof of Proposition~\ref{prop:dimlp} is non-constructive since it relies on a probabilistic argument. It is actually not too difficult to make it constructive, as stated by Proposition~\ref{prop:comput-eps}, %We start by proving a tractable version of Proposition~\ref{prop:dimlp} 
when the linear program \eqref{eq:shrink} has integer entries.

\begin{proposition}\label{prop:comput-eps}
    Consider the statement of Proposition~\ref{prop:dimlp}. Suppose that the entries of the $c^i$, $b$, and $M$ are integer. Given a rational number $\alpha$ of polynomial size, rational numbers $\varepsilon_s$ smaller than $\alpha$ and satisfying the conclusion of the proposition can be computed in polynomial time.
\end{proposition}

The proof of Proposition~\ref{prop:comput-eps} uses the following easy lemma.

\begin{lemma}\label{lem:polynom}
    Let $P(X) = \sum_{j=0}^d a_j X^j $ be a non-zero polynomial of degree at most $d \ge 1$ in $\R[X]$, and let $\rho,\theta$ be positive numbers such that $\rho \le |a_j| \le \theta$ when $a_j \neq 0$. Then $P$ has no root in $(-\frac \rho {d\theta},\frac \rho {d\theta})\setminus\{0\}$.
\end{lemma}

\begin{proof}
    Let $x_0$ be a root of $P$ such that $x_0 \neq 0$. Denote by $j^*$ the smallest index such that $a_{j^*} \neq 0$. The relation $x_0 \neq 0$ implies $j^*\neq d$. Since $P(x_0)=0$, we have $|a_{j^*}| |x_0|^{j^*}\le \sum_{j=j^*+1}^d|a_j||x_0|^j$. If $|x_0| \ge 1$, then obviously $|x_0| \ge \frac \rho {d\theta}$. If $|x_0| < 1$, then $\rho |x_0|^{j^*} \le (d-j^*)\theta |x_0|^{j^*+1}$, which implies in turn that $|x_0| \ge \frac \rho {d\theta}$ again since $x_0 \neq 0$. 
\end{proof}

\begin{proof}[Proof of Proposition~\ref{prop:comput-eps}]
    Without loss of generality, we assume that $\alpha < 1$. We keep the same notation as in the proof of Proposition~\ref{prop:dimlp}. Remember in particular that $B$ and $S$ are disjoint subsets of $[d]$, with $B$ of size $r$ and $S$ of size $q$. Moreover, we set $Q(B,S) \coloneqq \det(M_B) P(B,S)$ and denote by $\|C\|_{\infty}$ and $\|M\|_{\infty}$ the maximum of the absolute values of the entries of respectively $C$ and $M$. (In case $C=0$---which is anyway an obvious case---we implicitly define $\|C\|_{\infty}=1$.) Notice that the matrix $Q(B,S)$ has integer entries, whose absolute values are upper bounded by $(r+1)! \, \|C\|_{\infty} \, \|M\|_{\infty}^r$.

    The goal is to build explicitly $\lambda(B,S)$ as in the proof of Proposition~\ref{prop:dimlp} and an $\varepsilon \in \R_+^d$ such that the following two properties hold:
    \begin{enumerate}[label=(\alph*)]
        \item\label{norm} $\|\varepsilon\|_{\infty} < \alpha$,
        \item\label{non-int} $\lambda(B,S)^\top \bigl(M_S^\top M_B^{-1\top} \varepsilon_B - \varepsilon_S\bigl)$ is not $1$ when $P(B,S)$ is non-singular, and is not $0$ when $P(B,S)$ is singular.
    \end{enumerate}
    With such $\lambda(B,S)$ and $\varepsilon$, there cannot be any $\overline{S}$ of size $q$ satisfying~\eqref{eq:reducedsbar}, which translates into the fact that the optimal face has dimension at most $q-1$.

  For each pair $(B,S)$ (as in the proof of Proposition~\ref{prop:dimlp}) for which $P(B,S)$ is non-singular, set $\mu(B,S) \coloneqq \bigl(\com(Q(B,S))\bigl)^\top \boldsymbol{1}$, where the operator $\com$ applied to a matrix returns its cofactor. For each pair $(B,S)$ for which $P(B,S)$ is singular, proceeds as follows: choose a non-singular square submatrix $Q'(B,S)$ of $Q(B,S)$ of maximal rank; denote respectively by $R \subset [q]$ and $S' \subset S$ the rows and columns chosen this way from $Q(B,S)$; let $f \in \Z^R$ be the vector obtained by summing the columns of $Q(B,S)$ in $S \setminus S'$, and by keeping only the entries in $R$; set $\mu(B,S) \in \R^S$ to be
 the vector $-\bigl(\com(Q'(B,S))\bigl)^\top f$ completed with entries equal to $\det(Q'(B,S))$. Notice that the absolute values of the entries of $\mu(B,S)$ are upper bounded by
 \[
 q! \, \bigl((r+1)! \, \|C\|_{\infty} \, \|M\|_{\infty}^r \bigl)^{q-1} = q! \, \bigl((r+1)!\bigl)^{q-1} \, \|C\|_{\infty}^{q-1} \, \|M\|_{\infty}^{r(q-1)}\, .
 \]
When $P(B,S)$ is non-singular, set $\lambda(B,S) \coloneqq \frac{\det(M_B)}{\det(Q(B,S))}\mu(B,S)$. In this case, we have $\lambda(B,S) =P(B,S)^{-1}\boldsymbol{1}$ and
    \[
\lambda(B,S)^\top \bigl(M_S^\top M_B^{-1\top} \varepsilon_B - \varepsilon_S\bigl) \neq 1 \quad \Longleftrightarrow \quad
\det(M_B) \mu(B,S)^\top (M_S^\top M_B^{-1\top}\varepsilon_B - \varepsilon_S\bigl) \neq \det(Q(B,S)) \, .
\]
When $P(B,S)$ is singular, set $\lambda(B,S) \coloneqq \det(M_B)\mu(B,S)$. In this case, we have $\lambda(B,S)$ non-zero and
    \[
\lambda(B,S)^\top \bigl(M_S^\top M_B^{-1\top} \varepsilon_B - \varepsilon_S\bigl) \neq 0 \quad \Longleftrightarrow \quad
\det(M_B) \mu(B,S)^\top (M_S^\top M_B^{-1\top}\varepsilon_B - \varepsilon_S\bigl) \neq 0 \, .
\]
Let us check that in this case we have in addition $P(B,S) \lambda(B,S) = 0$. The entries in $R$ of $Q(B,S)\mu(B,S)$ are by construction equal to $-Q'(B,S)\bigl(\com(Q'(B,S))\bigl)^\top f + \det(Q'(B,S))f = 0$. Since the rank of the rows of $Q(B,S)$ in $R$ is the rank of $Q(B,S)$, every row not in $R$ is a linear combination of rows in $R$. This implies $Q(B,S)\mu(B,S) = 0$, and hence $P(B,S) \lambda(B,S) = 0$.

Write $\det(M_B) \mu(B,S)^\top (M_S^\top M_B^{-1\top}\varepsilon_B - \varepsilon_S\bigl)$ as $\sum_{s \in B \cup S} a_s \varepsilon_s$. The $a_s$ are integers and we have $\|a\|_{\infty} \le \theta'$, where we have chosen (not tight)
\[
\theta' \coloneqq (q+1)!\, \bigl((r+1)!\bigl)^q\, \|C\|_{\infty}^q\, \|M\|_{\infty}^{rq} \, .
\]
By Lemma~\ref{lem:polynom} with $\rho=1$, and $\theta = \max(2, \theta')$, setting $\varepsilon_s = \left(\frac {\alpha} {d\theta} \right)^s$ ensures that $\varepsilon$ satisfies the properties~\ref{norm} and~\ref{non-int}. (When $P(B,S)$ is non-singular, we use Lemma~\ref{lem:polynom} with $a_0 \coloneqq - \det(Q(B,S))$.) The number $\theta'$ is computable in polynomial time and so are the $\varepsilon_s$'s. 
\end{proof}

\section{Proof of Theorem~\ref{thm:main}}\label{sec:main}

%\todo[inline]{Frédéric will make a proposition for the organization of this section for the next meeting}

%\subsection{Main argument}

%\todo[inline]{Motivate a bit the introduction of the two optimization problems, and say something about $\varepsilon$ (e.g., we could add ``to be fixed later'')?}

%\todo[inline]{Use $L$}
We assume without loss of generality that $|S_h| = ns_h$ for every $h \in [k]$. Indeed, as noted right after the statement of Theorem~\ref{thm:main}, this can be safely done by adding dummy items. Let $L \coloneqq n(\sum_{i=1}^n\sum_{j=1}^m|u_i(j)|+1)$ and $\varepsilon_{ij} >0$ for all $i \in [n]$ and $j \in [m]$. 
For each point $(t_1,t_2,\ldots,t_n)$ of the $(n-1)$-dimensional standard simplex $\Delta^{n-1}$, we consider the problem 
\begin{equation}\label{prob}\tag{P$(t)$}
\begin{array}{rl}
  \text{maximize}   & \displaystyle{\frac 1 L \sum^n_{i=1}\bigl(1 + (L- n)t_i\bigl) u_{i}(A_i) + \sum_{i=1}^n\varepsilon_i(A_i)} \\[3ex]
   \text{s.t.}  & A \text{ is a feasible allocation,}
\end{array}
\end{equation}
where we use the notation $\varepsilon_i(S) = \sum_{j\in S}\varepsilon_{ij}$ for any subset $S$ of $[m]$.

We comment on the way~\eqref{prob} is written to provide some intuition or motivation. The first term the objective function is a weighted combination of the utilities of the agents: each agent $i$ gets a weight $t'_i \coloneqq \frac 1 L \bigl(1 + (L- n)t_i\bigl)$. These weights are non-negative and sum up to $1$ but never reach the value $0$: in other words, they are lying in the strict interior of $\Delta^{n-1}$, and actually in a shrunk copy of it. This is required to ensure that optimal solutions of~\eqref{prob} are Pareto-optimal; see the proof of Lemma~\ref{lem:po}. The second term of the objective is a perturbation that allows the application of Proposition~\ref{prop:dimlp}: with this proposition, we make sure that the dimension of the optimal faces of~\eqref{prob} is not too high and that we can apply Lemma~\ref{lem:fixed} below, which in turn controls the number of items common to all optimal allocations.

We also consider the linear relaxation of~\eqref{prob}. To ease its writing, we introduce the following bipartite graph $G=(V,E)$. The vertices on one side are all pairs $(i,h)$, where $i$ is an agent and $h$ is a category. The vertices on the other side are the items. We put an edge between a pair $(i,h)$ and an item $j$ if $j$ belongs to the category $h$. The linear relaxation of~\eqref{prob} can then be written as:
\begin{equation}\label{prob-relax}\tag{$\overline{\text{P}}(t)$}
\begin{array}{rll}
  \text{maximize}   &  \displaystyle{\frac 1 L \sum_{e=(i,h)j \in E}\bigl(1 + (L- n)t_i\bigl) u_{i}(j)x_e + \sum_{e =(i,h)j\in E}\varepsilon_{ij}x_e}\\[3ex]
   \text{s.t.}  & \displaystyle{\sum_{e \in \delta(j)} x_e = 1} & \forall j \in [m] \\[1.2ex]
   & \displaystyle{\sum_{e \in \delta(i,h)} x_e = s_h} & \forall i \in [n], h \in [k] \\[1.2ex]
   & x_e \ge 0 & \forall e \in E\, .
\end{array}
\end{equation}
Since the constraint matrix of~\eqref{prob-relax} is totally unimodular, every optimal solution of~\eqref{prob} is also optimal for~\eqref{prob-relax}.

We sketch now the proof of Theorem~\ref{thm:main} and the main lemmas to which it relies. We finish this subsection with the proof of Theorem~\ref{thm:main} with full detail. %The lemmas are proved in the subsequent subsection.

The proof of Theorem~\ref{thm:main} relies on the KKM (Knaster--Kuratowski--Mazurkiewicz) lemma~\citep{knaster1929beweis}. 

\begin{lemma}[KKM lemma]\label{lem:KKM}
Consider a collection $C_1,C_2,\ldots,C_{d+1}$ of closed subsets of the $d$-dimensional standard simplex $\Delta^d$ such that every $(t_1,t_2,\ldots,t_{d+1})$ in $\Delta^{d}$ belongs to $C_i$ for some $i$ with $t_i > 0$. Then there exists a point of $\Delta^d$ belonging to all $C_i$ simultaneously.
\end{lemma}

The KKM lemma is applied with $d=n-1$. The $C_i$ are defined with the help of optimal solutions of~\eqref{prob}: A point $(t_1,t_2,\ldots,t_n)$ belongs to $C_i$ if there exists an optimal solution of (\hyperref[prob]{P$({t})$}) that is envy-free for agent $i$. The fact that they satisfy the condition of the KKM lemma is a consequence of the following lemma.

\begin{lemma}\label{lem:ef}
    If $\varepsilon_{ij} < \frac 1 {Ln^2m}$ for all $i,j$, then every optimal solution of \eqref{prob} is envy-free for some agent $i^*$ with $t_{i^*}>0$.
\end{lemma}
\begin{proof}
    Consider an arbitrary optimal solution $A^*$ of (\hyperref[prob]{P$({t})$}) and assume for a contradiction that no agent $i$ with positive $t_i$ is envy-free. Define the {\em envy-graph} of $A^*$ whose vertices are the agents and where there is an arc from $i$ to $i'$ whenever agent $i$ envies agent $i'$. Since $A^*$ is Pareto-optimal (by Lemma~\ref{lem:po}), this graph does not admit a directed cycle. Pick $i_1$ with $t_{i_1} \ge 1 /n$. Take a directed path from $i_1$ to some sink of the envy-graph and denote by $i_1,i_2,\ldots,i_q$ the vertices of this path. Note that by assumption $t_{i_q}$ is $0$. Define a new allocation $A'$ as follows: 
\[
A'_i = 
\begin{cases}
A^*_{i_{\ell+1}}
& \text{if $i = i_{\ell}$,} \\
A_i^* & \text{otherwise.}
\end{cases} \,
\]
where $i_{q+1} = i_1$. The value of the objective function for $A'$ minus that for $A^*$ is equal to
\begin{equation}\label{eq:diff}
        \frac 1 L \sum^n_{i=1}\Bigl(1 + (L- n)t_i\Bigl) \Bigl(u_{i}(A'_i) - u_i(A^*_i)\Bigl) + \sum_{i=1}^n\left(\varepsilon_i(A'_i) - \varepsilon_i(A_i^*)\right) \, .
\end{equation}
The left-hand term can be expressed as
    \[
        \frac 1 L \sum_{\ell=1}^{q-1}\Bigl(1 + (L- n)t_{i_{\ell}}\Bigl) \Bigl(u_{i_{\ell}}(A^*_{i_{\ell+1}}) - u_{i_{\ell}}(A^*_{i_{\ell}})\Bigl) + \frac 1 L \Bigl(u_{i_q}(A^*_{i_1}) - u_{i_q}(A^*_{i_q})\Bigl)\, .
      \]  

%\todo[inline]{Do not have to change the following. Here,  $\Bigl(u_{i_q}(A^*_{i_1}) - u_{i_q}(A^*_{i_q})\Bigl) \geq - \sum_{i,j}|u_i(j)| \geq -\frac{1}{n} n (\sum_{i,j}(|u_i(j)|+1)-1)$}
Since the utilities are integers and agent $i_{\ell}$ envies agent $i_{\ell+1}$ when $\ell < q$, this left-hand term is lower-bounded by
\[
    \frac 1 L \sum_{\ell=1}^{q-1}\Bigl(1 + (L- n)t_{i_{\ell}}\Bigl) - \frac 1 L \frac {L-1} n   \ge  
        \frac 1 L \left(1 + (L- n) \frac 1 n \right) - \frac 1 L \frac {L-1} n = \frac 1 {L n} \, .
\]
All in all, the difference~\eqref{eq:diff} is larger than $0$, contradicting the optimaliy of $A^*$. 
\end{proof}

The point belonging to all $C_i$ simultaneously, whose existence is a consequence of the KKM lemma, translates into the existence of allocations $A^{(1)},A^{(2)},\ldots,A^{(n)}$, each being envy-free for a different agent. Their Pareto-optimality is a consequence of the following lemma.

%\todo[inline]{We want to see whether we get fractional Pareto-optimal.}

\begin{lemma}\label{lem:po}
    If $\varepsilon_{ij} < \frac 1 {Lnm}$ for all $i,j$, then every optimal solution of \eqref{prob} is Pareto-optimal.
\end{lemma}

\begin{proof}
    Consider an optimal solution $A^*$, and assume for contradiction it is not Pareto-optimal. Let $A'$ be some allocation that Pareto-dominates $A^*$. The value of the objective function for $A'$ minus that for $A^*$ is equal to
    \[
        \frac 1 L \sum^n_{i=1}\Bigl(1 + (L- n)t_i\Bigl) \Bigl(u_{i}(A'_i) - u_i(A^*_i)\Bigl) + \sum_{i=1}^n\left(\varepsilon_i(A'_i) - \varepsilon_i(A_i^*)\right) \, .
      \]  
      Since the utilities are integers, this quantity is lower bounded by 
      \[
      \frac 1 L + \sum_{i=1}^n\left(\varepsilon_i(A'_i) - \varepsilon_i(A_i^*)\right) > 0 \, ,
      \]
      contradicting the optimality of $A^*$. 
\end{proof}

The fact that there is a set of at most $\min \{k+1,n\}(n-1)$ items that can be reallocated so as to transform any allocation $A^{(i)}$ into an allocation $A^{(i')}$ with $i' \neq i$ results from Proposition~\ref{prop:dimlp} and the following lemma. Note that the following lemma concerns the linear relaxation of~\eqref{prob}, but we can draw conclusions for the allocations $A^{(i)}$ since, as noted above, they are also optimal for this linear relaxation. See Figure~\ref{fig:lemma} for an illustration of the lemma.

\begin{figure}[ht]
% =======================
% Utilities table
% =======================
\begin{minipage}[c]{0.48\textwidth}
\vspace{0pt}
\centering
\begin{tabular}{@{}l|llllll@{}} \toprule
{}   & $j_1$  & $j_2$ & $j_3$ & $j_4$& $j_5$ & $j_6$ \\ \midrule
agent $1$ & $1$ & $1$ & $1$  & $0$ & $1$ &$0$ \\ 
agent $2$ &  $1$ & $0$ & $0$  & $1$ & $0$ &$1$   \\ 
agent $3$ &  $0$ & $1$ & $1$  & $1$ & $0$ &$0$ \\ \bottomrule
\end{tabular}
\subcaption{Utilities $u_{ij}$ of agents for $6$ items.}
 \end{minipage}
  \hfill
% =======================
% Figure
% =======================
\begin{minipage}[c]{0.48\textwidth}
\vspace{0pt}
\begin{tikzpicture}[
  font=\small,
  v/.style={circle, draw, minimum size=7mm, inner sep=0pt},
  elab/.style={fill=white, inner sep=1.5pt}
]
% Left side (agents)
\node[v] (a1) at (0, 4) {$1$};
\node[v] (a2) at (0, 2) {$2$};
\node[v] (a3) at (0, 0.0) {$3$};

% Right side (items)
\node[v] (j1) at (3, 4) {$\dot{\jmath}_1$};
\node[v] (j2) at (3, 2) {$\dot{\jmath}_2$};
\node[v] (j3) at (3, 0.0) {$\dot{\jmath}_3$};
\node[v] (j4) at (3, -2) {$\dot{\jmath}_4$};
\node[v] (j5) at (-3, 4) {$\dot{\jmath}_5$};
\node[v] (j6) at (-3, 2) {$\dot{\jmath}_6$};

% Edges with labels 
\draw[thick,dotted] (a1) -- (j1); %1/3
\draw[thick,dotted] (a1) -- (j2); %1/3
\draw[thick,dotted] (a1) -- (j3); %1/3
\draw[ultra thick] (a1) -- (j5); %1

\draw[thick] (a2) -- (j1); %2/3
\draw[thick,dotted] (a2) -- (j4); %1/3
\draw[ultra thick] (a2) -- (j6); %1

\draw[thick] (a3) -- (j2); %2/3
\draw[thick] (a3) -- (j3); %2/3
\draw[thick] (a3) -- (j4); %2/3

%\node[fill=white] at (2.1,4) {$0.5$};
%\node[fill=white] at (3.1,3.1) {$0.5$};
%\node[fill=white] at (2.1,3) {$0.4$};
%\node[fill=white] at (3.1,2) {$0.6$};
%\node[fill=white] at (3.1,1) {$0.3$};
%\node[fill=white] at (2.1,1) {$0.3$};
%\node[fill=white] at (2.1,0) {$0.4$};

\end{tikzpicture}
\subcaption{An example of an optimal solution $(x_e)_{e\in E}$ of (\hyperref[prob-relax]{$\overline{\text{P}}(t)$}) for the instance illustrated on the left after perturbation. Thick edges represent edges of value $x_e=1$, thin edges represent edges of value $x_e=2/3$, and dotted edges represent edges of value $x_e=1/3$. The set $E_0$ consists of the eight edges $(1,j_1), (1,j_2), (1,j_3), (2,j_1), (2,j_4), (3,j_2), (3,j_3)$, and $(3,j_4)$.
}
\end{minipage}
\caption{Illustration of Lemma~\ref{lem:fixed}. There are $n=3$ agents, $m=6$ items, and $k=1$ category, where each agent must be allocated to exactly two items (i.e., $s_h=2$ for the single category $h$). Each agent has the same weight $t_i=1/3$ for $i=1,2,3$. By Proposition~\ref{prop:dimlp}, for sufficiently small $\varepsilon$, the corresponding (\hyperref[prob-relax]{$\overline{\text{P}}(t)$}) has an optimal face of dimension $2$. 
By Lemma~\ref{lem:fixed}, the number of variables $x_e$ that are fixed to $1$ is at least $m-\min \{k+1,n\}(n-1)=2$. 
%Consider a cycle $1 \rightarrow j_2 \rightarrow 3 \rightarrow j_4 \rightarrow 2 \rightarrow j_1 \rightarrow 1$. Increase $x_e$ by $\delta=1/3$ along the edges $(1,j_2), (3,j_4), (2,j_1)$ and decrease along the edges $(3,j_2), (2,j_4), (1,j_1)$ by the same $\delta = 1/3$ yields another optimal solution. In the resulting allocation, agent $1$ receives $2/3$ of item $j_2$, $1/3$ of item $j_3$, all of item $j_5$; agent $2$ receives items $j_1$ and $j_6$ entirely; and agent $3$ receives $1/3$ of item $j_2$, $2/3$ of item $j_3$, and all of $j_4$.   
}
\label{fig:lemma}
\end{figure}

\begin{lemma}\label{lem:fixed}
  For every $(n-1)$-dimensional face of the set of feasible solutions of~\eqref{prob-relax}, the number of variables $x_e$ that are fixed to $1$ is at least $m-\min \{k+1,n\}(n-1)$.
\end{lemma}

\begin{proof}
    Denote by $E_0$ the set of edges $e$ of $G$ such that $x_{e}$ is not constant on the face, and by $r$ the number of items $j$ incident to $E_0$. 
    %\todo[inline]{Discussion: We argue about $m-r$ more often than $r$. Should we use $r$ to denote the number of vertices $j$ incident to $E_0$?}

%\todo[inline]{Change $K$. In the proof, agents also have degree at least two.}
    For each connected component $K$ of $G_0=(V,E_0)$ with at least one edge, we select an arbitrary spanning tree $T_K$. We denote by $n_K$ the number of vertices of the form $(i,h)$ in $K$ and by $r_K$ the number of items $j$ in $K$. Note that all vertices $(i,h)$ in a connected component $K$ share the same category $h$. 

    Pick an arbitrary point $x$ in the relative interior of the face. We have $0 < x_{e} < 1$ for each edge $e \in E_0$. Any arbitrary small change of an $x_{e}$ with $e \in E_0 \setminus \bigl(\bigcup_K E(T_K)\bigl)$ can be corrected by changing the values of $x_{e'}$ with $e' \in E(T_K)$ for some connected component $K$. 
    Indeed, $T_K$ together with edge $e$ admits a unique cycle $e_1,e_2,\ldots,e_{\ell}$ starting with $e_1=e$, which consists of an even number $\ell$ of edges since $G$ is bipartite. If $x_e$ decreases by $\delta$, we can adjust $x_{e_h}$ by increasing it by $\delta$ for even $h$ and decreasing it by $\delta$ for odd $h$, ensuring that the new $x$ remains within the same face.
    Thus, the dimension of the face is at least $|E_0|-\sum_K|E(T_K)|$. Since the face we consider is $(n-1)$-dimensional, we get
    \begin{equation}\label{eq:dim-forest}
        n-1 \ge |E_0|-\sum_K|E(T_K)| \, .
    \end{equation}

    From this inequality, we will derive two other inequalities. In what follows, we let $\kappa$ denote the number of connected components $K$ of $G_0=(V,E_0)$ with at least one edge. 
   
    %We have $n-1 \ge \kappa$, obtained from \eqref{eq:dim-forest} again, by noticing that the degree of every non-isolated vertex of $G_0$ is at least two (since $x_e$ is fractional for every $e \in E_0$), and thus that no connected component $K$ is a tree. 

   % Consider any connected component $K$ of $G_0=(V,E_0)$ with at least one edge. Note that $|E(T_K)| =n_K + r_K -1$ and $|E(K)| \ge 2r_K$ because the degree of every non-isolated vertex in $G_0$ is at least two. Combining these, we get that $|E(K)|- |E(T_K)| \ge r_K-n_K+1$, and thus $r_K$ is upper bounded by $|E(K)|-|E(T_K)| +n_K -1$. Summing this inequality over all connected components $K$ with at least one edge, we obtain 
  %  \begin{align*}
   % r=\sum_{K}r_K &\le |E_0|- \sum_{K}|E(T_K)| + \sum_{K}n_K - \kappa \\
   % &\le n-1 + \min \{nk,n(n-1)\} - 1 \\
   % &\le n \min \{k+1,n\} -2,
    %\end{align*}
%where the second inequality follows from \eqref{eq:dim-forest} and $\sum_{K}n_K \le \min \{kn,n \kappa\} \le \min \{kn,n (n-1)\}$ 

    Each component $K$ covers $r_K + n_K$ vertices, and hence $|E(T_K)| =r_K + n_K  - 1$.  Combining this with \eqref{eq:dim-forest}, we obtain
    $$
    |E_0| \le n-1 + \sum_{K} (r_K+n_K  - 1) =n-1 + r + \sum_{K} (n_K  - 1). 
    $$
    This is the first inequality. The second one is $\sum_{K}(n_K-1) \le \min \{k(n-1),(n-1)^2\}$. Indeed, from \eqref{eq:dim-forest} we obtain $n-1 \ge \kappa$. This follows from the observation that every non-isolated vertex of $G_0$ has degree at least two (since $x_e$ is fractional for every $e \in E_0$), and thus that no connected component $K$ is a tree. Consequently, $\sum_{K}(n_K-1) \le \kappa (n-1) \le (n-1)^2$. Moreover, we have $\sum_{K}(n_K-1) \le k(n-1)$ since for each category $h$, the sum of the values $n_K-1$ over all connected components $K$ containing vertices $(i,h)$ associated with that category is at most $n-1$. Combining these two inequalities leads to $|E_0| \le (n-1) +r +\min \{k(n-1),(n-1)^2\}= r+ \min \{k+1,n\} (n-1)$.
    On the other hand, $|E_0| \ge 2r$ because, as already noted, the degree of every non-isolated vertex in $G_0$ is at least two. Combining this last inequality with the previous one leads then to $r \le  \min \{k+1,n\} (n-1)$. Since each item $j$ not incident to $E_0$ provides exactly one edge $e$ such that $x_{e}$ is fixed to $1$ on the face, we get the desired conclusion. 
\end{proof}

We provide now the proof in full detail.

\begin{proof}[Proof of Theorem~\ref{thm:main}]
As explained above, we start by choosing a suitable value for the $\varepsilon_{ij}$. We do it with the help of Proposition~\ref{prop:dimlp}, with $c_e^i \coloneqq u_i(j)$ for $e=(i,h)j$ (where $h$ is the category of $j$), with  $M$ being the constraint matrix of (\hyperref[prob-relax]{$\overline{\text{P}}(t)$}), and with $b$ being the right-hand side of this linear program. This provides us with $\varepsilon_{ij} > 0$ for every $i \in [n]$ and every $j \in [m]$. Since these $\varepsilon_{ij}$ can be chosen arbitrary small, we can assume that they are all smaller than $\frac 1 {Kn^2m}$. Note that---and this will be used at the end of the proof---this choice of the parameters makes that the linear program~\eqref{eq:shrink} of Proposition~\ref{prop:dimlp} is exactly (\hyperref[prob-relax]{$\overline{\text{P}}(t)$}) when $t'_i = \frac{1}{K} (1+(K-n)t_i)$ for every $i \in [n]$.

%Since $t'_i= \frac{1}{K} (1+(K-n)t_i)$ for $i \in [n]$ are non-negative and sum up to $1$, we can apply Proposition~\ref{prop:dimlp} to~\eqref{prob-relax} with $c_e^i = u_i(j)$ for $e=(i,h)j$ (where $h$ is the category of $j$) and $M$ being the constraint matrix of (\hyperref[prob-relax]{$\overline{\text{P}}(t)$}), and $b$ its right-hand side. This provides us with $\varepsilon_{ij} > 0$ for every $i \in [n]$ and every $j \in [m]$. Since these $\varepsilon_{ij}$ can be chosen arbitrary small, we can assume that they are all smaller than $\frac 1 {Kn^2m}$. 

We consider now optimal solutions of~\eqref{prob}---with the $\varepsilon_{ij}$ we have just determined---for $t$ ranging over $\Delta^{n-1}$. We use these optimal solutions to define subsets $C_1,C_2,\ldots,C_n$ of $\Delta^{n-1}$: Let $C_i$ be the set of points $t \in \Delta^{n-1}$ for which there exists an optimal solution of (\hyperref[prob]{P$({t})$}) that is envy-free for agent $i$. We check that these $C_i$ satisfy the conditions of the KKM lemma (Lemma~\ref{lem:KKM}) with $d=n-1$. To see that $C_i$ is closed, take a sequence of points $t^{\ell}=(t^{\ell}_1,t^{\ell}_2,\ldots,t^{\ell}_n)_{\ell=1,2,\ldots}$ in $C_i$ that converges to some ${\bar t}$. For each $\ell$, there exists an optimal solution $A^{\ell}$ of the problem (\hyperref[prob]{P$(t^{\ell})$})
such that $A^{\ell}$ is envy-free for agent $i$. Since there are finitely many allocations, up to taking a subsequence of the original sequence, we can assume that $A^{\ell}$ is the same allocation ${\bar A}$. This allocation ${\bar A}$ is also optimal for (\hyperref[prob]{P$({\bar t})$}) because the objective function is continuous in $t$. By the choice of the $\varepsilon_{ij}$'s, Lemma~\ref{lem:ef} implies that the sets $C_i$ satisfies the covering condition.

By the KKM lemma (Lemma~\ref{lem:KKM}), there exists $t^* \in \Delta^{n-1}$ belonging to all $C_i$ simultaneously. Denote by $A^{(1)},A^{(2)},\ldots,A^{(n)}$ the optimal solutions of (\hyperref[prob]{P$({t^*}$})) that are respectively envy-free for agents $1,2,\ldots,n$. The $\varepsilon_{ij}$'s were chosen according to Proposition~\ref{prop:dimlp}. This implies that the set of optimal solutions of the linear program~\eqref{eq:shrink} with $t'_i= \frac 1 K \bigl(1 + (K- n)t_i^*\bigl)$ for every $i \in [n]$---the other parameters being set at the beginning of the proof---is a face of dimension at most $n-1$. In other words, the set of optimal solutions of (\hyperref[prob-relax]{$\overline{\text{P}}(t^*)$}) is a face of dimension at most $n-1$.
%Set $\tilde t_i\coloneqq \frac 1 K \bigl(1 + (K- n)t_i^*\bigl)$ for all $i \in [n]$. The $\varepsilon_{ij}$'s were chosen according to Proposition~\ref{prop:dimlp} with $t'_i$, which implies that the set of optimal solutions of (\hyperref[prob-relax]{$\overline{\text{P}}(t^*)$}) is a face of dimension at most $n-1$. 
The constraints of (\hyperref[prob-relax]{$\overline{\text{P}}(t^*)$}) being totally unimodular, the allocations $A^{(1)},A^{(2)},\ldots,A^{(n)}$ correspond to optimal solutions of this linear program. Lemma~\ref{lem:fixed} implies then that $m-\min\{k+1,n\}(n-1)$ items are allocated the same way in these allocations, and Lemma~\ref{lem:po} shows that they are all Pareto-optimal. 
\end{proof}

\section{Proof of Theorem~\ref{thm:algo}}\label{sec:algo}

%\todo[inline]{Rewrite this section to parallel Section~\ref{sec:algo}}

In this section, we prove Theorem~\ref{thm:algo}. We note that the enumeration step in our algorithm is carried out in the polar space (namely, the space of objective functions). This is different from the approach in \cite{barman2025mixedmanna}, where the enumeration is performed in the decision space.
Now we return to the graph $G$ and the linear programming \eqref{prob-relax} defined in Section~\ref{sec:main}.
For every elementary cycle $C$ of $G$, we introduce a hyperplane $H_C$ of $\R^E$, defined as follows. Write $C$ as its sequence $e_1,e_2,\ldots,e_{\ell}$ of edges, with arbitrary first edge. Note that $\ell$ is even as $G$ bipartite. Then $H_C \coloneqq \{y \in \R^E \colon y_{e_1} - y_{e_2} + y_{e_3} - y_{e_4} + \cdots - y_{e_{\ell}} = 0\}$. We consider the arrangement formed by these hyperplanes $H_C$ when $C$ ranges over the set of all elementary cycles of $G$. 

\begin{lemma}\label{lem:costs}
    Consider~\eqref{prob-relax}, but with a general objective function of the form $\sum_{e \in E} c_ex_e$, for any $c\in \R^E$.  Given two costs $c^1,c^2$, if the supporting cell of $c^1$ contains that of $c^2$, then all optimal solutions for $c^1$ are optimal for $c^2$.
\end{lemma}

\begin{proof}
    We rely on the approach described by Schrijver for transportation problems~\cite[Chapter 21]{schrijver2003combinatorial}. Given a cost $c$ and a feasible solution $x$, we introduce a directed graph $D_x=(V,A)$ with arc costs, where $V$ is still the vertex set of $G$, but where $A$ is built as follows. For each edge $e=(i,h)j$ of $G$, there is an arc $\bigl((i,h),j\bigl)$ with cost $-c_e$ if $x_e > 0$, and an arc $\bigl(j,(i,h)\bigl)$ with cost $c_e$ if $x_e < 1$. The following characterization holds:
    {\em $x$ is optimal if and only if there is no directed cycle of total positive cost.} (This is Theorem 21.12 in the aforementioned book. Note that in the book it is a minimization problem.) Note that we can ignore directed cycles of length two in that characterization since there are all always of zero cost.

    %Take now two costs $c,c'$ in the relative interior of a same cell, and 
    Assume that the supporting cell of $c^1$ contains that of $c^2$. Consider any optimal solution $x^\star$ for $c^1$. Let us show that $x^\star$ is also optimal for $c^2$. Consider a directed cycle $\vec{C}$ in $D_{x^\star}$ of length at least four. By optimality, its cost with respect to $c^1$ is non-positive. The total cost of the directed cycle writes then as $c_{e_1}^1 - c_{e_2}^1 + \cdots - c_{e_{\ell}}^1$, which means that $c^1$ is located on the non-positive side of $H_C$, where we keep the notation $C$ for the undirected underlying cycle. Since $c^2$ is located on the same side as $c^1$, we have $c^2_{e_1} - c^2_{e_2} + \cdots -  c^2_{e_{\ell}} \le 0$, which means that the total cost of $\vec{C}$ with respect to $c^2$ is non-positive as well. 
\end{proof}

The proof of Theorem~\ref{thm:algo} relies on the following elementary facts from linear algebra, which we state and prove for sake of completeness.

\begin{lemma}\label{lem:inters-aff}
    The intersection of two affine subspaces of $\R^d$ respectively of dimension $d_1$ and $d_2$ is either empty, or of dimension at least $d_1 + d_2 - d$.
\end{lemma}

\begin{proof}
Consider two affine subspaces of $\R^d$, with a non-empty intersection. Pick a point $x$ in the intersection. The two subspaces can then be written under the form $x + F_1$ and $x + F_2$, where $F_1$ and $F_2$ are two linear subspaces of $\R^d$, with $\dim(F_1) = d_1$ and $\dim(F_2) = d_2$. The intersection of the affine subspaces can then be written as $x + F_1 \cap F_2$. By basic linear algebra, we have $\dim(F_1 \cap F_2) = \dim(F_1) + \dim(F_2) - \dim (F_1 + F_2)$. The dimension of $F_1+F_2$ being upper-bounded by $d$, we get that the dimension of $F_1 \cap F_2$ is lower bounded by $d_1+d_2-d$, which implies the desired conclusion. 
\end{proof}

\begin{lemma}\label{lem:subspace-arrangement}
    Consider an affine subspace $F$ determined by the intersection of some hyperplanes in $\R^d$. Then $F$ is already determined by exactly $d-\dim(F)$ of them.
\end{lemma}

\begin{proof}
    Let $H_1, H_2, \ldots, H_r$ be hyperplanes such that $F = \bigcap_{\ell=1}^r H_{\ell}$. We prove now by descending induction on $\dim(F) = d'$ that there is a subset $S\subseteq [r]$ of cardinality $d-d'$ such that $F=\bigcap_{\ell \in S}H_{\ell}$.

    Suppose first that $\dim(F) =d-1$. In this case, $F$ is a hyperplane. It must be contained in every $H_{\ell}$ and for dimensional reason coincides actually with every $H_{\ell}$. Any $H_{\ell}$ already determines $F$. Suppose now that the result is true for some $d' \le d-1$. Assume that $\dim(F)=d'-1$. Let $r' \le r$ be the minimal integer such that $F = \bigcap_{\ell = 1}^{r'} H_{\ell}$. By Lemma~\ref{lem:inters-aff} with $F_1 = \bigcap_{\ell = 1}^{r'-1} H_{\ell}$ and $F_2 = H_{r'}$, the dimension of $\bigcap_{\ell = 1}^{r'-1} H_{\ell}$ is at most $d'$. By the minimality of $r'$, it is at least $d'$, and so it is exactly $d'$. By induction, there is a subset $T \subseteq [r'-1]$ of cardinality $d-d'$ that $H_1 \cap \cdots \cap H_{r'-1} = \bigcap_{\ell \in T}H_{\ell}$. Setting $S = T \cup \{r'\}$ allows to conclude the claim. 
\end{proof}

\begin{proof}[Proof of Theorem~\ref{thm:algo}]
For each $i \in [n-1]$ and $e \in E$, let $c^i_e = u_i(j)$ if $e$ corresponds to an edge $(i,h)j$. Up to multiplying the $u_i(j)$ by some integer, we can assume that the $c_e^i$ are integers, and we can thus define some $\varepsilon$ according to Proposition~\ref{prop:comput-eps}, with $M$ being the constraint matrix of (\hyperref[prob-relax]{$\overline{\text{P}}(t)$}), $b$ its right-hand side, and $\alpha = \frac 1 {Kn^2m}$.

Let $Q$ be the set of vectors $(c_e(t))_{e \in E}$ with $c_e(t) = 
\frac 1 K \bigl(1 + (K- n)t_i\bigl)c^i_e + \varepsilon_e$ (still with $e=(i,h)j$) and $t$ in the $(n-1)$-dimensional standard simplex $\Delta^{n-1}$. Note that these vectors are exactly the cost vectors of (\hyperref[prob-relax]{$\overline{\text{P}}(t)$}). The hyperplanes $H_C$ determine an arrangement within the polytope $Q$: the intersection of any $H_C$ with $Q$ is of dimension $\dim(Q)$ or $\dim(Q) - 1$ (by Lemma~\ref{lem:inters-aff}); keep only those for which the intersection is of dimension $\dim(Q)-1$; these intersections are then hyperplanes of the affine hull of $Q$.
 
Consider now a vertex $(c_e(t))_{e \in E}$ of this arrangement. Remark defining $t'_i \coloneqq \frac 1 K \bigl(1 + (K- n)t_i\bigl)$ makes that $t'$ belongs to the $(n-1)$-dimensional simplex $\Delta^{n-1}$.  Proposition~\ref{prop:comput-eps} and Lemma~\ref{lem:fixed} together show that for the corresponding optimal face, there are at least $m-\min\{k+1,n\}(n-1)$ items $j$ for which the value of $x_e$, with $e$ incident to $j$, is fixed to $1$. We determine these edges by fixing, for each edge $e$ independently, $x_e$ respectively to $0$ and $1$, and see whether the optimal value remains the same by solving the corresponding linear program (which can be done in polynomial time). If the optimal value remains the same, then $x_e$ is not fixed on the optimal face. Otherwise, it is fixed. For each vertex of the arrangement, there are thus at most $\min\{k+1,n\}(n-1)$ items that are not allocated in a unique way and we try all possible allocations of these items by brute-force. According to Lemma~\ref{lem:costs}, we do not lose any optimal solution, and in particular, we will consider the optimal solutions $A^{(1)},A^{(2)},\ldots,A^{(n)}$ of (\hyperref[prob]{P$({t^*}$})) in the proof of Theorem~\ref{thm:main}, which are Pareto-optimal and envy-free respectively for agents $1,2,\ldots,n$. For each vertex of the arrangement, this makes $O(n^{\min\{k+1,n\}(n-1)})$ possible integral optimal allocations to check. 
%\todo[inline]{Change $O(n^{n(n-1)})$}
    
By Lemma~\ref{lem:subspace-arrangement}, each vertex is determined by at most $n-1$ hyperplanes of the affine hull of $Q$. (Note that such a vertex can be located on a face of $Q$, but then Lemma~\ref{lem:subspace-arrangement} shows that less than $n-1$ hyperplanes determine the vertex.) Each vertex is thus determined by at most $n-1$ hyperplanes $H_C$. Since there are $O(m^{2n})$ cycles $C$, this makes $O(m^{2n(n-1)})$ vertices, and we get the desired complexity.
\end{proof}

\section{Concluding remarks}
Theorem~\ref{thm:main} guarantees that envy-freeness for each agent can be achieved  after reallocating at most $\min\{k+1,n\}(n-1)$ items. 
An open question is whether this bound can be improved, allowing for envy-freeness with fewer reallocations. We note in Proposition~\ref{prop:lowerbound} below that the lower bound is $n$.
Additionally, in Theorem~\ref{thm:algo}, we assume that the number $n$ of agents is a constant. It would be interesting to explore whether a more efficient algorithm exists in terms of computational complexity.

\begin{proposition}\label{prop:lowerbound}
For every number $n$ of agents, there exists an instance such that, for every Pareto-optimal allocation, there is an agent for whom at least $n$ items must be reallocated in order to obtain another Pareto-optimal allocation in which that agent is not envious.
\end{proposition}
\begin{proof}
Consider an instance with $n$ agents $1,2,\ldots,n$, $n$ items $j_1,j_2,\ldots,j_n$, and a single category $h$ with $s_h=1$. For each agent $i=2,3,\ldots,n-1$, agent $i$ has utility $1$ for items $j_i$ and $j_{i+1}$, and utility $0$ for all other items. Agent $1$ has utility $10$ for $j_1$, utility $1$ for $j_2$, and utility $0$ for all other items. Agent $n$ has utility $10$ for $j_1$, utility $1$ for $j_n$, and utility $0$ for all other items.

Let $A$ denote the allocation that assigns item $j_i$ to agent $i$ for each $i=1,2,\ldots,n$, and let $B$ denote the allocation that assigns item $j_{i+1}$ to agent $i$ for each $i=1,2,\ldots,n$, where $j_{n+1}=j_1$.

We first show that in any Pareto-optimal allocation, every agent must receive positive utility. Suppose there is a feasible allocation $X$ in which some agent receives utility $0$. Since each agent receives at most one item and there are exactly $n$ items and $n$ agents, every feasible allocation must in fact assign every agent to exactly one item. If agent $n$ does not receive item $j_1$ in $X$, then under $A$, agent $1$ receives utility $10$, agent $n$ receives utility $1$, and every agent $i=2,\ldots,n-1$ receives utility $1$. Hence every agent weakly prefers $A$ to $X$, and the agent with utility $0$ in $X$ is strictly better off under $A$. Thus $A$ Pareto-dominates $X$. On the other hand, if agent $n$ does receive item $j_1$ in $X$, then under $B$, agent $n$ receives utility $10$, agent $1$ receives utility $1$, and every agent $i=2,\ldots,n-1$ receives utility $1$. Again, every agent weakly prefers $B$ to $X$, and the agent with utility $0$ in $X$ is strictly better off under $B$. Thus $B$ Pareto-dominates $X$.
Therefore, every Pareto-optimal allocation gives positive utility to every agent. Since every feasible allocation must assign every agent to exactly one item, there are only two Pareto-optimal allocations, $A$ and $B$.

Finally, in allocation $A$, agent $n$ receives utility $1$ from $j_n$, while agent $1$ receives $j_1$, which agent $n$ values at $10$. Thus agent $n$ envies agent $1$. Eliminating this envy while preserving Pareto optimality requires moving from $A$ to $B$, because $A$ and $B$ are the only Pareto-optimal allocations. Since $A$ and $B$ differ on all $n$ items, this requires reallocating all $n$ items.
\end{proof}

\section*{Acknowledgments}
This work was partially supported by JST FOREST Grant Numbers JPMJFR226O.

\bibliographystyle{plainnat}
%\bibliography{abb,cake}

\end{document}